\documentclass[12pt,reqno,a4paper]{amsart}

\usepackage[margin=3cm,footskip=1cm]{geometry}

\usepackage{amssymb} \usepackage{amsmath} \usepackage{amsfonts}
\usepackage{amsthm} \usepackage{mathtools}
\usepackage{esint}

\usepackage{enumerate} \usepackage[latin1]{inputenc}
\usepackage[shortlabels]{enumitem}

\usepackage{color}
\usepackage{graphicx} \usepackage{verbatim}

\newcommand{\Ran}{{\operatorname{Ran}}}

\newcommand{\cs}{{\rm the Cauchy-Schwarz inequality }}

\newcommand{\Opw}{{\operatorname{Op\!^w\!}}}
\newcommand{\N}{{\mathbb{N}}} 
\newcommand{\R}{{\mathbb{R}}} 
\newcommand{\C}{{\mathbb{C}}}

\newcommand{\e}{{\rm e}}

  \renewcommand{\i}{{\rm i}}

\renewcommand{\Re}{{\rm Re}\,} \renewcommand{\Im}{{\rm Im}\,}

\DeclarePairedDelimiter\inp\langle\rangle

\newcommand\parb[2][]{#1 \big ( #2#1\big )}

\newcommand{\mand}{\text{ and }}

\newcommand{\vE}{{\mathcal E}}

\theoremstyle{plain}
\newtheorem{thm}{Theorem}[section]
\newtheorem{proposition}[thm]{Proposition}
 \newtheorem{corollary}[thm]{Corollary}
\theoremstyle{definition}

 \newtheorem{remark}[thm]{Remark}
\newtheorem{remarks}[thm]{Remarks}

\newtheorem*{remarks*}{Remarks}
\newtheorem*{remark*}{Remark}
\newtheorem*{defn*}{Definition}

\numberwithin{equation}{section}

\title{Decay of eigenfunctions of elliptic PDE's, II}

\author{I. Herbst}
\address[I.  Herbst]{Department of Mathematics \\
  University of Virginia \\
  Charlottesville \\
  VA 22904\\ U.S.A.}
\email{iwh@virginia.edu}

  \author{E. Skibsted} \address[E. Skibsted]{Institut for Matematiske
  Fag \\
  Aarhus Universitet\\ Ny Munkegade 8000 Aarhus C, Denmark}
\email{skibsted@math.au.dk}

\begin{document}

  \begin{abstract} We study exponential  decay rates of eigenfunctions of self-adjoint
    higher order elliptic operators on  $\R^d$.  We are interested in decay rates as a function of direction.  We show that the
    possible  decay rates are to a large extent determined
    algebraically. 
\end{abstract}


\maketitle

\tableofcontents

\section{Introduction and previous results}\label{subsec:result}
Consider  a real elliptic polynomial $Q$ of degree $q$ on  $\R^d$.  ($Q$ elliptic means that for large $\xi \in \R^d, C|Q(\xi)| > |\xi|^q$ for some $C$.)
We consider  the operator
$H=Q(p)+V(x)$, $p=-\i \nabla$, on $L^2(\R^d)$ with  $V$
bounded
and  measurable. For most of our results we assume
 $\lim_{|x| \to \infty} V(x) = 0$ and additional decay properties of
 the potential.  By the assumptions on $Q$
the operator $Q(p)$ is self-adjoint with domain the standard Sobolev space of
order $q$ which consequently is also the domain of $H$. The goal of the
 paper is
to study exponential decay of $L^2$-eigenfunctions of $H$ with eigenvalue $\lambda \in \mathbb{R}$ as a function of direction. It is the
second in a series of two papers on exponential decay.   The first one is \cite{HS}.

In \cite{Ag1}, Agmon investigated the asymptotic behavior of the
Green's function (the integral kernel of the inverse of $Q(p) -
\lambda$ for spectral parameter $\lambda$ in the resolvent set of
$Q(p)$).  In certain cases he obtained rather precise asymptotics of
this function.  Since we are investigating the asymptotic behavior of
eigenfunctions of $Q(p)+V(x)$ with $V(x)$ small at infinity, one might
suspect that the asymptotic behavior of the Green's function would
determine the exponential rate of fall-off of the eigenfunction.  This
is false in a rather spectacular way: First, the eigenvalue $\lambda$
may actually be in the spectrum of $Q(p)$ where the Green's function
decays (at most) like an inverse power of $|x|$ while the
eigenfunction decays exponentially.  And second, whether or not the
eigenvalue is in the spectrum of $Q(p)$, there may be 
several (global or local)  decay rates which
occur for different potentials $V$ of compact support.  Of course at
least one of these decay rates will not reflect the asymptotic
behavior of the Green's function.  Already in \cite{HS} we gave
examples of these phenomena.  For another example see Section
\ref{sec:example}.  These phenomena do not occur if
  $Q(\xi) = |\xi|^2$, at least if for example $V = o(|x|^{-1/2})$ at
  infinity (see Theorems \ref{thm:global decay rate determined} and
  \ref{rotationinv}).

We first summarize some of the results of \cite{HS} which will be our starting  point.  References to previous work are given there.  We define
the {\it global decay rate}  of $\phi \in L^2(\mathbb{R}^d)$ as
  \begin{equation}\label{eqn:sigmac}
    \sigma_g=\sup\{\sigma\geq 0|\e^{\sigma|x|}\phi\in L^2\}.
  \end{equation}\ 

It is intuitively clear that $\sigma_g$ is determined by  the directions of weakest exponential decay of $\phi$. \\ 

  In the rest of this section we assume that $(H-\lambda)\phi = 0$ with $\lambda \in \mathbb{R}$ and $\phi \in L^2(\mathbb{R}^d)$.  We will mostly assume there is a splitting of $V,$  $V=V_1+V_2$, into bounded functions, with $V_1$
smooth and real-valued and $V_2$  measurable, with additional assumptions depending on
the result.

\begin{thm}\label{thm:start}
Under either of the following two conditions we can conclude that $\sigma_g >0$:
\begin{enumerate}[1)]
\item\label{item:1}
$\lambda \notin {\Ran }Q: = \{Q(\xi)|\xi \in \mathbb{R}^d\}$  and $V(x) = o(1)$ at infinity.
\item\label{item:2}
$\lambda  \in {\Ran }Q$ but $\lambda$ is not a critical value of
$Q(\xi)$, $\xi$ real, and in addition
\begin{align*}
&\forall \alpha: \partial^{\alpha}V_1(x) = o(|x|^{- |\alpha|}),\\
&V_2(x) = o(|x|^{-1}).
\end{align*}
\end{enumerate}
\end{thm}

Earlier work for the Laplacian can be found in \cite{Oc, CT, FH, MP1}.
Carleman type estimates which can be useful in proving part ~\ref{item:2} of
Theorem \ref{thm:start} for even more general operators
were proved in \cite{MP2}.

 The following theorem eliminates the possibility of super-exponential decay at the expense of rather strong decay assumptions on the potential:

\begin{thm}\label{thm:super}
Suppose $V_2(x) = O(|x|^{-q/2 -\delta})$ and
$\partial^{\alpha}V_1(x) = O(|x|^{-(\delta+q +|\alpha|)/2})$,  $1\le
|\alpha| \le q$, where $\delta > 0$.  Then
  $\sigma_g<\infty$ unless
  $\phi=0$.

\end{thm}

For $Q(\xi) = |\xi|^2$ or $|\xi|^4$ (and perhaps for any real elliptic $Q$)  one can do with weaker decay assumptions on $V$, see \cite{HS}. In fact for $Q(\xi) = |\xi|^2$ or $|\xi|^4$, in the conditions on $V, q $ can be replaced by $q/2$.  (Of course the results given in \cite{HS} for the Laplacian were known, see \cite{BM, FHH2O1, FHH2O2, FHH2O3, FH}.)

With the above two theorems we have conditions on $V$ which guarantee that $0 < \sigma_g < \infty$.  We will assume the latter in the rest of this paper.

The next theorem shows that  $\sigma_g$ must satisfy certain equations
which in  favorable situations  determine its possible values.

\begin{thm}\label{thm:global decay rate determined}

Suppose $0 < \sigma_g < \infty$.  If
\begin{align*}
&\forall \alpha: \partial^{\alpha}V_1(x) = o(|x|^{- |\alpha|}),\\
&V_2(x) = o(|x|^{-1/2}),
\end{align*}
then there exists $(\omega, \xi, \beta) \in S^{d-1}\times \mathbb{R}^d \times \mathbb{C}$ such that
 \begin{subequations}
  \begin{align}\label{eq:1}
    Q(\xi+\i \sigma_g\omega)&=\lambda,\\
    \nabla_\xi Q(\xi +\i \sigma_g \omega)&= \beta \omega.\label{eq:2}
  \end{align}
  \end{subequations}
\end{thm}

Note that the number of real unknowns indicated by  $\sigma_g, \omega,
\xi, \beta$ equals the number of real equations in (\ref{eq:1}) and
(\ref{eq:2}) and thus the set of  $\sigma_g$ occurring as solutions of
these equations has a chance of being discrete. In fact except for a
finite set of  exceptional $\lambda$'s this is true if $Q$ is
rotationally  invariant (one might say in spite of the rotation invariance).   In \cite{HS} it is shown that  except possibly for this finite set of $\lambda$'s every solution $\sigma_g > 0$ of these equations actually occurs for a real, smooth $V$ of compact support.

In this paper we will study quantities  somewhat similar  to
\eqref{eqn:sigmac}. One of those is a rough measure of   the asymptotics
 at infinity. It is the {\it local decay
  rate}  of  any  
$\phi\in L^2$ defined for $\omega \in S^{d-1}$ by 
\begin{align}\label{eq:loca}
  \sigma_{loc}(\omega) = \sup\{\sigma | \e^{\sigma |x| }\phi \in L^2(C)
  \ \text{for some open cone C containing} \ \omega \}.
\end{align} 

In the next section we introduce in addition two other measures of
exponential rate of decay which also depend on direction.  Those
notions appear more amenable to analysis than the local decay rate,
but as we will see our study of these other notions of decay yields
information on $\sigma_{loc}(\omega) $. Our main result will be
presented in Section \ref{sec:calc-decay-rate}, see Theorem
\ref{rademacher1}. It allows us to some extent to calculate rates of
decay of eigenfunctions in $L^2$, most notably for  
  rotationally invariant $Q$'s, see Theorem \ref{rotationinv}
  (announced earlier in \cite{HS}). This is in the spirit of Theorem
\ref{thm:global decay rate determined}, that is by solving a certain
system of algebraic equations. We give another demonstration of our
results for an example in Section \ref{sec:example} (a
non-rotationally invariant case). In Section \ref{sec:Agmon} we
elaborate on a connection to previous works \cite{Ag1, Ag2}.  We show
how the above mentioned system of algebraic equations relates to
\cite{Ag1, Ag2} and in fact, more generally, can be derived
by a variational principle.  In Subsection \ref{Green's function} we discuss the exponential decay of the Green's function when the spectral parameter is outside $\text{Ran}Q$.  Finally we
have collected various considerations on possible smoothness of rates
of decay of eigenfunctions in Section \ref{sec:set-barmathcale}.

  \section{Directional decay rates, arbitrary
    $\phi$}\label{sec:direct-decay-rates}

In this section $\phi$ is an arbitrary function in $L^2(\mathbb{R}^d)$ 
with $0 < \sigma_g < \infty$ where  $\sigma_g$ is defined in (\ref{eqn:sigmac}).  Note that we do not assume that $\phi$
is an eigenfunction. The basic object which incorporates information
on the directional decay rates of $\phi$ and which we find most amenable to analysis is the set
\begin{equation}\label{eq:E}
  \mathcal{E} = \{\eta \in \mathbb{R}^d | \e^{\eta \cdot x}\phi \in L^2\}.
\end{equation}
We introduce three exponential decay rates depending on a direction  $\omega \in S^{d-1}$.
\begin{align*}
&\sigma_c(\omega) = \sup\{\sigma | \e^{\sigma \omega \cdot x} \phi \in L^2\}\\
&\sigma_s(\omega) = \sup\{\eta \cdot \omega | \eta \in \mathcal{E}\} \\
&\sigma_{loc}(\omega) = \sup\{\sigma | \e^{\sigma |x| }\phi \in L^2(C) \  \text{for some open cone C containing} \ \omega \}
\end{align*}

It is easy to see that
\begin{align*}
  \sigma_g \le \sigma_c(\omega) \le \sigma_s (\omega) \le \sigma_{loc}(\omega).
\end{align*}

Note that $\sigma_s$, as the supremum of a family of continuous
functions, is lower semi-continuous.  In addition if we define
$\sigma_s(t\omega)  =  t \sigma_s(\omega) $ for $t\ge 0$, then
$\sigma_s(x)$ is the support function of the set $\mathcal{E}$ (by definition $0\cdot \infty = 0$).

Here are some basic facts which are true for an arbitrary $\phi \in
L^2$ if $\sigma_g \in (0, \infty)$.  We allow $\sigma_c(\omega) =
\infty$ in which case we define $1/\sigma_c(\omega) = 0$. Since
$\sigma_g<\infty$ a simple compactness argument shows that $\sigma_c(\omega) <
\infty$ for at least one $\omega$, in fact $\sigma_g =\inf_{\omega}\sigma_{loc}(\omega)$. By $B_r(x)$ we mean the open
  ball in $\R^d$ of radius $r$ centered at $x$.

\begin{thm} \label{3sigmas}
\begin{enumerate}[1)]
\item \label{item:1a} $\mathcal{E}\ \text{is convex}$ and contains
  $B_{\sigma_g}(0)$.
\item
$1/\sigma_c(\omega)\text{ is Lipschitz.  In fact}\
|1/\sigma_c(\omega_1) - 1/\sigma_c(\omega_2)|  \le |\omega_1 -
\omega_2|/\sigma_g.$ In particular the set $\{\omega \in S^{d-1}|
\sigma_c(\omega) < \infty\}$ is a relatively open subset of $S^{d-1}$.
\item
$\partial \mathcal{E} = \{\sigma_c(\omega)\omega | \omega \in S^{d-1}, \sigma_c(\omega) < \infty\}.$
\item
Suppose $f : \mathbb{R}^d \rightarrow [0, \infty)$ is convex and $f(tx) = tf(x)$ for all $t\ge 0$.   Suppose in addition 
\begin{align*}
 \e^{tf}\phi \in L^2\ \text{for all} \ t < 1.
 \end{align*}
 Then $f(x) \le \sigma_s(x)$.
\item
The function $\sigma_{loc}$ is lower semi-continuous.  Suppose
$\rho:\mathbb{R}^d \to [0,\infty)$ is continuous, $\rho(tx) =
t\rho(x)$ for all $t\ge 0$ and $\rho(\omega) \le \sigma_{loc}(\omega)$
for all $\omega\in S^
{d-1}$. Then
\begin{align*}
 \e^{t\rho}\phi \in L^2\ \text{for all} \ t < 1.
 \end{align*}
\end{enumerate}
\end{thm}
\begin{proof}
{\bf 1) }
Take $\eta_j \in \mathcal{E}, j= 1,2$.  By the Young inequality for
any $s\in(0,1)$
\begin{align*}
  \e^{(s\eta_1+(1-s)\eta_2)\cdot x}\leq s\e^{\eta_1\cdot
    x}+(1-s)\e^{\eta_2\cdot x}.
\end{align*} This implies that $s\eta_1+(1-s)\eta_2\in
\mathcal{E}$. Similarly by \cs $B_{\sigma_g}(0)\subset \vE$.

{\bf 2) }
Given  $\omega_1$, $\omega_2$, $\omega_1\neq\omega_2$,  and $\mu \in(0, \sigma_c(\omega_1))$  we will choose
$\sigma$ so that $\sigma \omega_2 \in \mathcal{E}$ as follows.  We write
$$\sigma \omega_2 = \sigma \omega_1 + \sigma(\omega_2 - \omega_1),$$
and define $\eta_1 = \sigma \omega_1/(1-t)$ and $\eta_2 =
\sigma(\omega_2 - \omega_1)/t$ so that $\sigma \omega_2 = (1-t) \eta_1
+ t\eta_2.$ In order to have $ t \in (0,1)$, $\eta_1 \in \mathcal{E}$
and $\eta_2 \in \mathcal{E}$ (so that \ref{item:1a} applies), we demand
$0 < \sigma/(1-t) \le \mu$ and $\sigma|\omega_2 - \omega_1|/t <
\sigma_g$.  We choose $t$ so that $t \sigma_g/|\omega_2 - \omega_1| =
(1-t)\mu$.  We find that if $1/\sigma > 1/\mu + |\omega_2 -
\omega_1|/\sigma_g$ then indeed $\sigma \omega_2 \in \mathcal{E}$.  It
follows that $1/\sigma_c(\omega_2) \le 1/\sigma_c(\omega_1) +
|\omega_2 - \omega_1|/\sigma_g$.  Interchanging $\omega_2$ and
$\omega_1$ gives the result.

{\bf 3) } We first show that if $\eta \in \bar{\mathcal{E}}$ then
$s\eta \in \text{int}(\mathcal{E})$ for $0<s<1$.  Pick a sequence
$\eta_j \in \mathcal{E}$ with $\eta_j \rightarrow \eta$.  By
\ref{item:1a} $s\eta_j + (1-s)\zeta \in \mathcal{E}$ if
$|\zeta| < \sigma_g$.  This means $s\eta_j + B_{(1-s)\sigma_g}(0)
\subset \mathcal{E}$.  Since $s\eta_j \rightarrow s\eta$, $s\eta \in
B_{(1-s)\sigma_g}(s \eta_j) \subset \mathcal{E}$ for large enough $j$.
Whence indeed $s\eta \in \text{int}(\mathcal{E})$. We next show that if $\eta \in \partial \mathcal{E}$ then with $ \eta
= |\eta| \omega, \omega \in S^{d-1}$, we have $\sigma_c(\omega) =
|\eta|$.  Since $\eta \in \bar{\mathcal{E}}, s\eta \in
\text{int}(\mathcal{E})$ for $0<s<1$.  Thus $\sigma_c(\omega) \ge
|\eta|$.  Suppose $\sigma_c(\omega) > |\eta|$.  Fix $\sigma \in
(|\eta|, \sigma_c(\omega))$.  Then $\sigma \omega \in
\text{int}(\mathcal{E})$.  Fix $r> 0$ so that $\sigma \omega + B_r(0)
\subset \mathcal{E}$.  For $0 < s< 1$ we have $(1-t)s\eta + t (\sigma
\omega + B_r (0) \subset \mathcal{E}$ for any $t \in [0,1]$.  Let $t =
\frac{|\eta| - s|\eta|}{\sigma - s|\eta|}$.  Then $(1-t)s\eta + t
\sigma \omega = \eta$ so that $\eta + B_{tr}(0) \subset \mathcal{E}$
and thus $\eta$ is not a boundary point.  We have shown
$\sigma_c(\omega) = |\eta|$ or in other words $\sigma_c(\omega) <
\infty$ and $\eta = \sigma_c(\omega) \omega$.  If on the other hand
$\sigma_c(\omega) < \infty$ and $\eta = \sigma_c(\omega) \omega$, then
$(1-n^{-1})\sigma_c(\omega) \omega \in \mathcal{E}$ for all $n \in
\mathbb{N}$ by the definition of $\sigma_c(\omega)$.  Hence
$\sigma_c(\omega) \omega \in \bar{\mathcal{E}}$.  $\sigma_c(\omega)
\omega$ cannot be in $ \text{int}(\mathcal{E})$ by the definition of
$\sigma_c(\omega)$ so that $\sigma_c(\omega) \omega \in \partial
\mathcal{E}$.

{\bf 4) }
Since $f$ is convex, for each $x_0 \in \mathbb{R}^d$ there exists a
set of linear functions, $l_{\eta,x_0}(x) = f(x_0) + \eta \cdot
(x-x_0), \eta \in G(x_0)$ (here $G(x_0)$ is our notation for the set
of subgradients at $x_0$, see for
  example \cite[p. 214]{Ro}), so that
\begin{equation*}
f(x) = \sup\{l_{\eta,x_0}(x)| x_0 \in \mathbb{R}^d, \eta \in G(x_0)\}.
\end{equation*}

Using $t^{-1}f(tx) = f(x)$ we have $f(x) \ge  \eta \cdot x +t^{-1} l_{\eta, x_0}(0)$ for  $x_0 \in \mathbb{R}^d, \eta \in G(x_0)$  Taking $t$ to infinity we obtain $f(x) \ge \sup\{ \eta \cdot x|  \eta \in G\}$
where $G = \cup_{x_{0} \in \mathbb{R}^d} G(x_0)$.  But since $0 =f(0)
\ge l_{\eta, x_0}(0)$ we have  $f(x) = \sup\{ \eta \cdot x + l_{\eta,
  x_0}(0)| x_0 \in \mathbb{R}^d, \eta \in G(x_0)\} \le \sup\{ \eta
\cdot x| x_0 \in \mathbb{R}^d, \eta \in G(x_0)\}$.  Thus $ f(x) =
\sup\{\eta\cdot x| \eta \in G\}$.  This means that if
  $\eta \in G$ then $\e^{t\eta \cdot x} \phi \in L^2$ for all $t \in [0,1)$ so that $G \subset \bar{ \mathcal{E}}$.  Thus
\begin{equation*}
f(x) \le \sup\{ \eta \cdot x| \eta \in \bar{ \mathcal{E}} \} = \sigma_s(x).
\end{equation*} 

{\bf 5) }
Given $\alpha \in \mathbb{R}$ the set $\{\omega \in S^{d-1}| \sigma_{loc}(\omega) > \alpha\}$ is open by definition of $ \sigma_{loc}$ (allowing $\sigma_{loc}(\omega) = \infty$).  Thus again by definition,  $\sigma_{loc}$ is lower semi-continuous.
Suppose $\sigma_{loc}(\omega_0) = \infty$.  By the continuity of
$\rho$ we can find an open cone $C_{\omega_0}$ containing $\omega_0$
and $\sigma_0 \in (0, \infty)$ so that  if $\omega \in S^{d-1} \cap
C_{\omega_0}$ then $\rho(\omega) < \sigma_0$.  By shrinking
$C_{\omega_0}$ if necessary, we can assume $\e^{\sigma_0|x|} \phi \in
L^2(C_{\omega_0})$.  Thus $\e^{t\rho} \phi \in  L^2(C_{\omega_0})$ for
all $t \in [0,1]$.  If  $\sigma_{loc}(\omega_0) < \infty$, given $t
\in [0,1)$  we can find $\sigma_0$ such that $t\sigma_{loc}(\omega_0) < \sigma_0$ with $\e^{\sigma_0 |x|}\phi \in L^2(C_{\omega_0})$ where $C_{\omega_0}$ is an open cone containing $\omega_0$.  By assumption $t \rho(\omega_0) < \sigma_0$.  Thus by continuity there  is a smaller open cone $\tilde C_{\omega_0} \ni \omega_0$ so that  $t \rho(\omega) < \sigma_0$ for $\omega \in \tilde C_{\omega_0}$ which implies $\e^{t\rho}\phi \in L^2(\tilde C_{\omega_0})$.  The result then follows by the compactness of $S^{d-1}$ and a covering argument.

\end{proof}

As we will see in the next section, if $\phi$ is an eigenfunction of $H = Q(p) + V(x)$ with eigenvalue $\lambda$, under favorable conditions we will be able to calculate the possible values of  $\sigma_s(\omega)$ from our knowledge of $Q(\xi)$ and the eigenvalue $\lambda$.  We do not have a direct method of calculating $\sigma_{loc}(\omega)$.  Thus it is important to know when $\sigma_{loc}(\omega) = \sigma_s(\omega)$.

We call the (affine) hyperplane, $0=(\eta - \eta_0)\cdot \omega_0$, with parameters $(\omega_0, \eta_0) \in S^{d-1}\times \partial \mathcal{E}$,  a supporting hyperplane if $(\eta - \eta_0)\cdot \omega_0 \le 0$ for all $\eta \in \bar {\mathcal{E}}$.  Every point $\eta_0 \in \partial \mathcal{E}$ has at least one supporting hyperplane (\cite {Ro}, p.100).  Note that by definition of $\sigma_s$, if the hyperplane with parameters $(\omega_0, \eta_0)$ is a supporting hyperplane, $\sigma_s(\omega_0) = \eta_0 \cdot \omega_0$.  If there is a unique supporting hyperplane passing through $\eta_0 \in \partial \mathcal{E}$ we call $\eta_0$ a \emph {regular point} of $\partial \mathcal{E}$.   Otherwise  we refer to $\eta_0 \in \partial \mathcal{E}$ as a \emph {singular point}.  If $\eta_0$ is a regular point then $\partial \mathcal{E}$, parametrized by $\eta = \sigma_c(\omega) \omega$  with $\omega \in S^{d-1}$, is differentiable at $\eta_0$.  (Using the coordinates of some plane through the origin of dimension $d-1$, $\partial \mathcal{E}$ can be written as the graph of a convex function $f$.  The function $f$ is differentiable at a point $x_0$ if and only if $f$ has a unique subgradient at $x_0$ (\cite{Ro}, p. 242). This is the same as saying that $\bar{\mathcal{E}}$ has a unique supporting hyperplane at $(x_0, f(x_0))$.  Note that from Theorem \ref{3sigmas}, $\partial\mathcal{E}$ is Lipschitz, so that by Rademacher's theorem it is given locally by a function differentiable almost everywhere.)  Note also that if all points in $\partial \mathcal{E}$ are regular, then $\partial \mathcal{E}$ is $C^1$ (\cite {Ro}, p. 246).

\begin{thm} \label{regularpoints}
Suppose $\omega_0 \in S^{d-1}$ is given so that for \emph{some} regular point $\eta_0 \in \partial \mathcal{E}$  the hyperplane with parameters $(\omega_0, \eta_0)$ is a supporting hyperplane.  Then
\begin{equation}\label{k=loc}
\sigma_{loc}(\omega_0) = \sigma_s(\omega_0)=\eta_0\cdot \omega_0.
\end{equation}
\end{thm}

\begin{proof}
We have $\sigma_s(\omega_0) = \eta_0 \cdot \omega_0$.  Suppose
$\sigma_{loc}(\omega_0) > \sigma_s(\omega_0) $, then we will obtain a
contradiction.  First note that  $\e^{t_0 \eta_0 \cdot \omega_0|x|}\phi \in L^2(C_{\omega_0})$ for some open cone $C_{\omega_0}$ containing $\omega_0$ and some $t_0> 1$.  The continuity of $ \omega \mapsto \eta_0 \cdot \omega$ implies that by choosing $t_0 > 1$  smaller if necessary we can assume $\e^{t_0 \eta_0 \cdot x}\phi \in L^2(C_{\omega_0})$.

Now let  $\theta \in S^{d-1}$ with $\theta \ne \omega_0$.  By
definition of $\sigma_s(\theta)$ we have $\eta_0 \cdot \theta \le
\sigma_s(\theta)$.  If we have equality, then both of the hyperplanes with
parameters $(\omega_0, \eta_0)$ and $(\theta, \eta_0)$ are supporting at
$\eta_0$ contradicting the assumption that  $\eta_0$ is a  regular
point. Thus  $\eta_0 \cdot \theta \ < \sigma_s(\theta)$ for all $\theta \ne \omega_0$.

Given a unit vector $\theta$ in the complement of $C_{\omega_0}$, it follows that there is an $\eta \in \mathcal{E}  $ so that $\eta_0 \cdot \theta <  \eta \cdot \theta \le \sigma_s(\theta)$. Since $\e^{\eta\cdot x}\phi \in L^2$ there is an open cone $C_{\theta}$ containing $\theta$ and a $t_{\theta} > 1$ such that $\e^{t_{\theta}\eta_0\cdot x}\phi \in L^2(C_{\theta})$.   Hence by a covering argument $\e^{u\eta_0\cdot x}\phi \in L^2(\mathbb{R}^d)$ with $u > 1$  contradicting the fact that $ \eta_0 \in \partial \mathcal{E}$.

\end{proof}

\begin{corollary}\label{C1}
Suppose $\mathcal{E}$ is bounded and $\partial \mathcal{E}$ is $C^1$.  Then $\sigma_{loc}(\omega)=\sigma_s(\omega) $ for all $\omega \in S^{d-1}$.
\end{corollary}

\begin{remark} \ 
\begin{enumerate}[1)]
\item
Notice the emphasis on the word ``some" in Theorem \ref{regularpoints}.  The point $\eta_0 \in \partial \mathcal{E}$ in that theorem may not be unique and there may be singular points and regular points which all satisfy $\eta_0 \cdot \omega_0 = \sigma_s(\omega_0)$.  It is easy to show that if $\bar{\mathcal{E}}$ is strictly convex
then $\sigma_s(\omega_0) = \eta(\omega_0)\cdot \omega_0$ for a unique $\eta(\omega_0) \in \partial \mathcal{E}$. However we will have no need to assume strict convexity.
\item 
See Section \ref{sec:example} for an operator $H=Q(p)+V(x)$ and a
corresponding eigenfunction with  a real eigenvalue 
$\lambda\not\in\Ran Q$ such that 
the     assumption of Theorem \ref{regularpoints} is  fulfilled for some values
of $\omega_0$ while for other values of $\omega_0$ the conclusion of
the theorem is false, that is $\sigma_{loc}(\omega) > \sigma_s(\omega)$ for some $\omega$.

\end{enumerate}
\end{remark}

\section{Calculating the decay rate, $H\phi = \lambda
  \phi$}\label{sec:calc-decay-rate}

In this section we assume that $\phi$ is an eigenfunction of $H$ with eigenvalue $\lambda$.  We assume that the global decay rate, $\sigma_g$, of $\phi$ is positive.  We cannot completely eliminate the possibility that for some $\omega, \sigma_c(\omega) = \infty$ (unless $d=1$), but the next result limits the size of the set where this might occur.  See Theorem \ref{rademacher1} for a very different result which under unrelated assumptions shows $\sigma_c(\omega)  < \infty$ for all $\omega \in S^{d-1}$.  

\begin{proposition}\label{nosuper}
If $d=1$ and  $\phi \ne 0, \sigma_c(\pm 1) < \infty$ as long as $V$ is bounded.  If $d \ge 2$, under the hypotheses of Theorem \ref{thm:super} the set $\{\omega \in S^{d-1}| \sigma_c(\omega) = \infty\}$ lies in a hyperplane containing $0$ unless $\phi = 0$.
  \end{proposition}

\begin{proof}

We use the notation $\langle x\rangle = (|x|^2 + 1)^{1/2}$.
If $d=1$ suppose $\sigma_c(1) =\infty$.  Let $\phi_{\sigma}(x) =
\e^{\sigma x}\phi(x)$.  Then $(Q(p + \i\sigma) - \lambda)\phi_{\sigma} =
-V\phi_{\sigma}$.  $Q(z)-\lambda$ has finitely many zeros so that
$\lim_{\sigma \rightarrow \infty} \|(Q(p + \i\sigma) - \lambda)^{-1}
V\| =0 $. Since for large $\sigma,  \phi_{\sigma} = -(Q(p + \i\sigma) -
\lambda)^{-1} V\phi_{\sigma}$, we obtain $\phi_{\sigma} = 0$. Suppose
$d \ge 2$ and that $\sigma_c(\omega_j) = \infty$ for a set of linearly
independent vectors $\omega_1, \cdots, \omega_d$.  Since $\mathcal{E}$
is convex, there is an open cone $C$ so that  $ \sigma_c(\omega) =
\infty$ for $\omega \in C\cap
S^{d-1}$.  Choose $\omega_0 \in C \cap S^{d-1}$.  Then for small enough $\delta > 0$, if $f(x) = \delta r + \omega_0\cdot x$, we have $\e^{\sigma f}\phi \in L^2$ for all $\sigma > 0$.  Here we take $r= r_{\epsilon}$, $r_{\epsilon} = \langle x\rangle - \langle x\rangle^{1-\epsilon} + 1$ as in \cite{HS}, because of its good convexity properties.  The parameter $\epsilon > 0$ will be taken very small at the end of the proof.
As in \cite{HS}, we let $\phi_{\sigma  } = \e^{\sigma f} \phi$ and
$a = p - \i\sigma \nabla f(x)$ and note that $(Q(a^*) + V_1 -
\lambda)\phi_{\sigma} = -V_2\phi_{\sigma}$.  Taking norms of both
sides of this equation gives
\begin{align*}
  \inp{\phi_{\sigma}, ([Q(a),Q(a^*)] + |Q(a) + V_1 -\lambda|^2 )\phi_{\sigma}}=\inp{\phi_{\sigma}, (2\Re[V_1, Q(a)] +|V_2|^2)\phi_{\sigma}}.
\end{align*}
The only properties of $a$ and $a^*$ which were used to prove Theorem 1.4 in \cite {HS} are the form of the commutator $[a_j, a^*_k]$ and the form of $[a_j, V_1]$ which are virtually the same in the present situation:  $p_{jk} = [a_j, a^*_k] = 2\sigma \delta \partial_j\partial_k r$ (and thus after a calculation $(p_{jk}) \ge c \sigma r^{-1 - \epsilon}$).  Similarly $[a_j, V_1] = -\i\partial_jV_1$ is the same as in \cite{HS}.  Thus the proof of Theorem 1.4 in \cite{HS} works exactly in the same way to give the desired result after $\epsilon$ is chosen small enough (see \cite{HS}). 
\end{proof}

Note that if $\sigma_c(\omega_0) = \infty$, then $\sigma_{loc}(\omega) = \infty$ in the open half sphere $\{\omega \in S^{d-1}| \omega \cdot \omega_0 > 0\}$.  We have not found examples of this phenomenon in the case where $\phi$ is an eigenfunction of $H = Q(p) + V(x) $ with $V(x) = o(1)$ at infinity and $\sigma_g < \infty$.  

We now embark on a program to calculate the possibilities for
$\sigma_c(\cdot)$. We assume as above that $\phi$ is an eigenfunction
of $H$ with eigenvalue $\lambda$ and   $ \sigma_g >0$.  Our first result can put some restrictions on the pairs $(\lambda, \sigma_c(\omega))$.

\begin{proposition}\label{just energy}
Suppose $\omega_0 \in S^{d-1}$ with $\sigma_c(\omega_0) < \infty$.   Suppose $V=o(1)$ at infinity.  Then for some $\xi \in \mathbb{R}^d$
\begin{equation} \label{energyeqn}
Q(\xi + \i \sigma_c(\omega_0)\omega_0) = \lambda. 
\end{equation}
  \end{proposition}
\begin{proof} Abbreviate $ \sigma_c(\omega_0) = \sigma_0$  and use $r = \langle x \rangle$. For  $\epsilon>0$ 
we consider
\begin{align*}
  f(x)=(\sigma_0-\epsilon)\omega_0 \cdot x  + 2
\epsilon  r\mand f_n = (\sigma_0-\epsilon)\omega_0 \cdot x  + 2
\epsilon r/(1+r/n),\,n\in\N,
\end{align*}
 and $\phi_n = \e^{f_n} \phi$.  We will
show that unless (\ref{energyeqn}) is satisfied for some $\xi$,
$\|\phi_n\| \le C$ with a constant $C$ independent of $n$ provided
$\epsilon>0$ is chosen small enough. Taking $n\to \infty$ yields
$\e^{f } \phi \in L^2$ which is contradiction since $f(x)\geq
(\sigma_0+\epsilon)\omega_0 \cdot x$.

  We introduce the notation of \cite{HS}
  \begin{align*}
    X = \text
{Re}(Q(\xi + \i\nabla f_n(x)) - \lambda)\mand  Y =  \text {Im}Q(\xi + \i
\nabla f_n(x)).
  \end{align*}
  Suppose (\ref{energyeqn}) does not have a
solution. Then by a continuity and compactness 
 argument and the fact that $|\nabla f_n(x) - \sigma_0 \omega_0| \leq
 3\epsilon$,  we obtain
 \begin{align*}
   X^2 + Y^2 = |Q(\xi
+ \i \nabla f_n(x)) - \lambda |^2 \ge 2 \kappa\text{ for some small  }\kappa
>0.
 \end{align*}
  Obviously here we  needed $\epsilon>0$ small.  

Next we use the localization symbols $\chi_-=\chi(X^2+Y^2\leq
\kappa)$ and $\chi_+=\chi(X^2+Y^2\geq
\kappa)$  of \cite{HS} as well as their quantizations $\tilde
\chi_\mp$, respectively.  Here $\chi(t\le \kappa) =
  \chi_1(t/\kappa),\  \chi(t\ge \kappa) = \chi_2(t/\kappa)$, where 
  $\chi_1,\chi_2$ denote smooth non-negative 
functions with $\chi_1(t) = 1$ for $t \le 1$, $\chi_2(t) = 1$ for $t\ge 2$, and $\chi_1^2 + \chi_2^2 = 1$.  By construction $\tilde \chi_-=0$, and whence
by 
\cite[(4.9)]{HS} we  have  $I\leq \tilde \chi^2_++ C/r^2$. Using the estimate $\|\tilde {\chi}_+\phi_n \|^2 \le C(\|V\phi_n\|^2 + \|r^{-1/2}\phi_n\|^2)$ from Lemma 4.3 of \cite {HS} we obtain 
\begin{align*}
   \|\phi_n\|^2 \le C(\|V\phi_n\|^2  + \|r^{-1/2}\phi_n\|^2), 
\end{align*}
which easily leads to $\|\phi_n\| \le C$ as desired.  Here we mention that
although \cite [Lemma 4.3]{HS} is stated only for
$f_n (x) = r (\sigma + \gamma/(1+r/n))$, for certain values of
$\sigma$ and $\gamma$, the proof given there works with minor
modifications for our $f_n$.
  \end{proof}

\begin{remark}

According to Propositions \ref{nosuper} and \ref{just energy} if $d=1, \sigma_g >0$, and $V(x) = o(1)$ at infinity, then the possible decay rates $ \sigma = \sigma_c(\pm1)$ can be calculated from the equation $Q (\xi + \i \sigma) = \lambda$. Note that the reality condition shows that the totality of decay rates calculated from $Q(\xi + \i \sigma) = \lambda$ at $+ \infty$ is the same as that at $-\infty$.  In fact it is easy to see that if $\sigma_1 $ and $\sigma_2 $ are two positive solutions to this equation then there is a (complex) smooth compactly supported $V$ and a smooth nonzero $\phi$ with decay rate $\sigma_1$ at $+ \infty$ and decay rate $\sigma_2$ at $- \infty$ such that $(Q(p) + V - \lambda)\phi = 0$.
\end{remark}

In the following we  assume $d \geq 2 $.

Our main result is the following theorem:

\begin{thm}\label{rademacher1} 
Suppose $(H - \lambda)\phi = 0, 0 < \sigma_g < \infty$ and $V$ satisfies
\begin{align}
&\forall \alpha: \ \partial^{\alpha}V_1(x) = o(|x|^{-|\alpha|}),  \label{assump1}\\ 
&V_2(x) = o(|x|^{-1/2}).  \label{assump2}
\end{align} 

For   $\omega_0\in
S^{d-1}$ with  $\sigma_0: =
\sigma_c(\omega_0) < \infty$ let $\eta_0=\sigma_0\omega_0$ and $\hat C_{0} = \{\hat x
\in S^{d-1} | \sigma_s(\hat x) = \eta_0 \cdot \hat
x\}$.  
  For any such  $\omega_0$ there exists  $(\xi, \theta, \beta) \in \mathbb{R}^d
\times \hat C_{0} \times \mathbb{C}$ solving the pair of equations
\begin{subequations}
\begin{align}\label{maineq1}
    Q(\xi+\i \eta_0)&=\lambda,\\
    \nabla Q(\xi +\i \eta_0)&= \beta \theta.\label{maineq2}
  \end{align}
  \end{subequations}

If the set of $\eta_0$'s  which occur in the set of all
solutions  $(\xi, \theta, \beta, \eta_0) \in
\mathbb{R}^d\times S^{d-1}\times \mathbb{C}\times \mathbb{R}^d$   to
the pair of equations 
(\ref{maineq1})  and  (\ref{maineq2}) is bounded, then $\sigma_c(\omega) < \infty$ for all $\omega \in S^{d-1}$.
\end{thm}

\begin{remarks}
\begin{enumerate}[1)]
\item
There may be spurious solutions to the system of equations
(\ref{maineq1}) and (\ref{maineq2}) which do not describe the
exponential decay of an eigenfunction.  This may happen for the finite
set of exceptional eigenvalues $\lambda$ which arises in rotationally
invariant $Q$ (see Theorem \ref{rotationinv} below) and it happens for the example in Section \ref{sec:example}.  Both of these problems can be (at least partially) traced to the fact that the spectral parameter $\lambda$ is a critical value of $Q$. It is well known that the set of critical values of $Q : \mathbb{C}^d \rightarrow \mathbb{C}$ is finite.  In fact the number of these critical values can be bounded by $(q-1)^d$ (see \cite{BR}). 
\item
 Assume $\lambda \in \mathbb{R}$ is not such a critical value.  Let us choose $\theta \in S^{d-1}$ and assume that there is a solution to the system of equations (\ref{maineq1}) and (\ref{maineq2}).  We are interested in the set of $\eta = \text{Im} z$  such that $\eta \cdot \theta$ is stationary with respect to variations of $z = \xi + \i \eta \in M = \{z| Q(z) = \lambda \}$.  The vectors $\nabla_{(\xi,\eta)}(\text{Re}Q)(\xi, \eta)$ and $\nabla_{(\xi,\eta)} (\text{Im}Q)(\xi, \eta)$  are linearly independent by the Cauchy-Riemann equations.  Introducing the Lagrange multipliers $\gamma_1$ and $\gamma_2$ and  setting the derivatives of $\eta \cdot \theta + \gamma_1\text{Re}Q(z) + \gamma_2\text{Im}Q(z)$ with respect to $\xi$ and $\eta$ equal to zero we find that in fact $\eta \cdot \theta$ is indeed stationary at a point $z = \xi + \i \eta$ which solves  (\ref{maineq1}) and (\ref{maineq2}) .  Given the existence of the set $\mathcal{E}$, the meaning of $\theta$ is that of a unit vector perpendicular to a supporting hyperplane to $\bar{\mathcal{E}}$ at the point $\eta \in \partial \mathcal{E}$.  Thus for $\eta' \in \bar{ \mathcal{E}}$, $\eta' \cdot \theta$ has a global maximum or minimum at the point $\eta' = \eta$.
 \item
 Consider the set $\mathcal{E}$ corresponding to an eigenfunction
 $\phi$ satisfying the assumptions of Theorem \ref{rademacher1}.  For
 each point $\eta$ in the boundary of $\mathcal{E}$ there must be a
 corresponding solution to  (\ref{maineq1}) and (\ref{maineq2}).  We
 must be able to put together a function $\eta = \sigma(\omega) \omega
 \ (\omega = \eta/|\eta|, \sigma(\omega) = |\eta|)$ from the
 (multiplicity of) solutions to (\ref{maineq1}) and (\ref{maineq2})
 which satisfies the requirements coming from the convexity of
 $\mathcal{E}$ and the (related) Lipschitz continuity of
 $1/\sigma(\omega)$. If there is no such function then there is no
 such eigenfunction (see Remark 1.6 (4)) in \cite{HS}).  And clearly
 if the only such functions $\sigma(\omega)$ are bounded then
 $\sigma_c(\omega) $ corresponding to $\phi$ must be bounded.  This
 generalizes a statement in Theorem \ref{rademacher1}.
\item Clearly Theorems \ref{thm:global decay rate determined} and
  \ref{rademacher1} have a similar nature. Their proofs are also
  similar (partly explaining  why the conditions on $V$ are the same)  although there are additional ideas necessary in the present
  paper. As noted in \cite{HS} the proof of Theorem \ref{thm:global
    decay rate determined} is rather robust and applies with
  modifications to 
  certain elliptic  variable coefficient differential operators and
  even certain pseudodifferential 
 operators
  with elliptic  symbol being uniformly real-analytic in the
  $\xi$-variable assuming $\sigma_g$ for the  given eigenfunction is smaller than the uniform
  analyticity radius, say denoted $\sigma_a$. The same can be said for Theorem
  \ref{rademacher1}  under the stronger condition
  $\sigma_c(\omega_0)<\sigma_a$ on  the
  eigenfunction. For example our proof works
  for  the symbol $(|\xi|^2+s^2)^{1/2}+V(x)$ assuming $0<
  \sigma_g\leq \sigma_c(\omega_0)<s$.
\end{enumerate} 
\end{remarks}

We defer the proof of Theorem \ref{rademacher1}.  We have the
following corollary for rotationally  invariant $Q$.

\begin{thm}\label{rotationinv}
Suppose $(H - \lambda)\phi = 0$, $V$ is as in Theorem
\ref{rademacher1}, and $0 < \sigma_g < \infty$. Suppose $Q$ is
rotation invariant.  Define the polynomial $G$ of degree $q/2$ so that
$G(\xi^2) = Q(\xi)$.   We assume all the zeros of $G - \lambda$ have
multiplicity one.  (There are at most $\frac{q}{2} -1$ values of
$\lambda$ for which this is not the case.) Then there are at most
$q/2$ positive numbers $\sigma_0$ (being independent of $\omega_0=\eta_0/|\eta_0|$)
for which there is a solution to the pair of equations
(\ref{maineq1}) and  (\ref{maineq2}) with $|\eta_0|=\sigma_0$, and $\sigma_g$ is one of them.  In addition, $\sigma_{loc}(\omega) = \sigma_g$ for all $\omega \in S^{d-1}$.
\end{thm}

\begin{proof}

From the rotation invariance,  (\ref{maineq1}) and  (\ref{maineq2}) reduce to $G(z\cdot z) = \lambda$ and $2G'(z\cdot z) z = \beta \theta, z = \xi + \i \sigma_0 \omega_0$.  Our assumptions imply $ z = \beta' \theta$ for some $\beta' \in \mathbb{C}$ and thus
we have $z = (\alpha + \i\sigma_0) \omega_0$ with $\alpha \in
\mathbb{R}$.  It follows that the set of $\sigma_0$'s which may occur is bounded.  In fact the set of such positive $\sigma_0$'s
consists of at most $q/2$ constants independent of $\omega_0$ and
according to Theorem \ref{thm:global decay rate determined} 
$\sigma_g$ is one of them.  From the continuity of $1/\sigma_c(\omega)$,
see Theorem \ref{3sigmas}, and the fact that $S^{d-1}$ is connected it
follows that $\sigma_c(\omega)=\sigma$ for some $\sigma\in (0,\infty)$
independent of $\omega$. Whence $B_\sigma(0)\subset\vE \subset \bar
B_\sigma(0)$, which in turn by Corollary 
\ref{C1} implies that $\sigma=\sigma_c(\omega) =\sigma_s(\omega) = \sigma_{loc}(\omega)$ and therefore that $\sigma_{loc}(\omega)=
\sigma_g$.


\end{proof}

The main work of this section is in the next proposition which needs
modified constructions defined as follows in terms of a large
parameter $m$: 

For a given $\phi\in L^2$ with $0 < \sigma_g < \infty$
and a given integer $m>1/\sigma_g$ we replace the quantities $\vE, \sigma_c$ and
$\sigma_s$ of Section \ref{sec:direct-decay-rates} by $\vE^m,
\sigma^m_c$ and $\sigma^m_s$, respectively, given by replacing $ L^2$
by $L_m^2(\R^d)=\e^{-r/m}L^2(\R^d)$ in the definitions in Section
\ref{sec:direct-decay-rates}. Here and henceforth $r = \langle x
\rangle$. Alternatively, this amounts to the old
  quantities with  $\phi$ replaced by $\e^{r/m}\phi$. Whence by   Theorem \ref{3sigmas} we obtain that $\vE^m$ is
convex containing some ball, $1/\sigma^m_c$ is Lipschitz and $\partial
\mathcal{E}^m = \{\sigma^m_c(\omega)\omega | \omega \in S^{d-1},
\sigma^m_c(\omega) < \infty\}$.  Moreover for any $\omega_0\in S^{d-1}$ with
$\sigma_0:=\sigma_c(\omega_0)<\infty$   we can bound 
{\begin{align}\label{eq:apprM}
  \tfrac 1m\leq \sigma_0-\sigma^m_c(\omega_0)\leq \tfrac 1m
  \tfrac{\sigma_0}{\sigma_g},
\end{align} which by Rademacher's theorem allows us to find
 a sequence
$\eta^m=\sigma^m_c(\omega^m)\omega^m$ of regular points in  $\partial
\mathcal{E}^m$ with $\eta^m\to \eta_0:=\sigma_0\omega_0$ for
$m\to \infty$. To obtain the second inequality in (\ref{eq:apprM}) note that if $0 < \sigma < \sigma_0(1-(m\sigma_g)^{-1})$ then if $\omega_0\cdot x/|x| \ge \sigma_g/\sigma_0$ we have $\sigma \omega_0 \cdot x/|x| + 1/m < \sigma_0 \omega_0 \cdot x/|x|$ while if $\omega_0 \cdot x/|x| < \sigma_g/\sigma_0$ then $\sigma\omega_0 \cdot x/|x| + 1/m < \sigma_g$.

\begin{proposition}\label{mainthm} Suppose $(H - \lambda)\phi = 0, 0 <
  \sigma_g < \infty$, $m>1/\sigma_g$, $\omega^m_0\in S^{d-1}$ with $\sigma^m_0: =
  \sigma^m_c(\omega^m_0) < \infty$  and that $V$ is as in Theorem
  \ref{rademacher1}.  Let $\eta^m_0=\sigma^m_0\omega^m_0$ and $\hat C^m_{0} = \{\hat x \in S^{d-1} |
  \sigma^m_s(\hat x) = \eta^m_0 \cdot \hat x\}$.  Suppose $\eta^m_0$ is a regular
  point of $\partial\vE^m$  so that $\hat C^m_{0}$
  consists of only one point, say $\theta^m_0$. Then there exists $(\xi,\beta) \in \mathbb{R}^d
\times \mathbb{C}$ solving the pair of equations
 \begin{subequations}
  \begin{align}\label{maineq:1} Q(\xi+\i (\eta^m_0 +\theta^m_0/m))&=\lambda,\\
\nabla Q(\xi +\i (\eta^m_0 +\theta^m_0/m))&= \beta \theta^m_0.\label{maineq:2}
  \end{align}
  \end{subequations}
   \end{proposition} 

\begin{proof} We drop the superscript $m$. So 
  fix $\omega_0 \in S^{d-1}$ with $\sigma_0 =
  \sigma_c(\omega_0) < \infty$.  
  
Let \begin{equation*}
\Delta_1 = \max \{|\nabla Q(\xi + \i  (\eta_0 +\theta_0/m))|^2  | \,\xi \in \mathbb{R}^d, Q(\xi + \i (\eta_0 +\theta_0/m)) = \lambda\}. 
\end{equation*} Note that indeed (\ref{maineq:1}) has a solution,
cf. Proposition \ref{just energy}.

Letting $P_{\perp}(\theta) u = u - (u\cdot \theta)\theta$ we introduce
\begin{equation*}
\delta_1 = \min \{|P_{\perp}(\theta_0)\nabla Q(\xi + \i (\eta_0 +\theta_0/m))|^2 | \,\xi \in \mathbb{R}^d,  Q(\xi + \i (\eta_0 +\theta_0/m)) = \lambda\}.
\end{equation*}
 We will  show that $\delta_1=0$ 
 proceeding   by the way of 
contradiction.  The contradiction if  $\delta_1>0$ will arise  by showing   
  that  $\e^{s \eta_0 \cdot x } \phi \in L_m^2(\mathbb{R}^d
 )$  for some $s>1$.  So suppose $\delta_1>0$.

{\bf Step I} (Construction of phases.)
Consider for (small) $\epsilon>0$
\begin{align*}
f(x)& = (\sigma_0-\epsilon)\omega_0 \cdot x +  r/m,\\
f_{ n}(x)& = f(x) + 2\epsilon  r/(1+r/n),\  n \in \mathbb{N},\\
F(x)& = f (x)+ 2\epsilon  r.
\end{align*}
 
Note that $\phi_n:=\e^{f_{ n} }\phi \in
 L^2(\mathbb{R}^d)$. 
 We will
show that 
$\|\phi_n\| \le K$ with a constant $K$ independent of $n$ provided
$\epsilon>0$ is chosen small enough. Taking $n\to \infty$ yields
$\e^{F } \phi \in L^2$ which is a contradiction since $F(x)\geq
(\sigma_0+\epsilon)\omega_0 \cdot x+  r/m$. A necessary smallness
condition on $\epsilon$ is
\begin{align}\label{convexity}
  \delta_1/m>2\epsilon\Delta_1.
\end{align}

   {\bf Step II} (Role of \eqref{convexity}, convexity.) Noting that $\partial_ i
r=x_i/r$ we can compute $\nabla f_n$ and then estimate
\begin{align*}
 |(\xi + \i (\eta_0+\tfrac{x/r}m) - (\xi + \i \nabla f_{n})| \leq 3\epsilon.  
\end{align*} If  $Q(\xi + \i
\nabla f_{ n}) \approx \lambda$ 
this will
  for small $\epsilon$ allow us to exploit the positivity
of $\delta_1$ in a phase-space argument. 
 More precisely we  claim that  there is an open
 cone $\tilde C_0 \supset
 C_{0}:=\R_+\theta_0=\{c\theta_0|\,c>0\}$ 
 so that the symbol
\begin{align}\label{eq:pois]}
  b_{n}(x,\xi)=\sum_{i,j} r\overline{\partial_i Q(\xi + \i\nabla f_{
      n})} (\partial_j \partial_i f_{ n})\partial_j Q(\xi + \i\nabla
  f_{ n})
\end{align}
has a positive lower bound for   $x \in \tilde C_0$ with $|x|\geq R$  and for $|Q(\xi + \i
\nabla f_{ n}) - \lambda|^2 \leq  2\kappa$  provided 
$R^{-1},\kappa,\epsilon >0$ are small enough. The bound
is uniform in    $n, x, \xi$.  This follows
from the computations
\begin{align*}
  \partial_j \partial_i (r/(1+r/n)) &= (1 +
r/n)^{-2}\partial_ j \partial_i r - 2 (r/n)(1+ r/n)^{-3} \partial_ j
r\partial_i r/r,\\
\partial_j \partial_i r&=(\delta_{ij}-x_jx_ir^{-2})/r.
\end{align*} Note that the non-convex part $- 4\epsilon (1/n)(1+ r/n)^{-3}
|x/r\rangle \langle x/r|$ of the Hessian $(\partial_j \partial_i
f_{ n})$ has the lower bound $-2\epsilon
 I/r$,  while the convex part    has the lower bound
$m^{-1} (I-|x/r\rangle \langle x/r|)/r$.  Whence for $x$ in a
small open  cone $\tilde C_0\supset C_{0}$ and $R^{-1},\kappa,\epsilon >0$
small    indeed we obtain a lower
bound of the above   form $ b_{n}(x,\xi)\geq c_1$ where the constant  $c_1$ can be
chosen as close to $ c_2:=-2\epsilon
 \Delta_1+ \delta_1/m$ as desired.
 The positivity of $ c_2$  is exactly
\eqref{convexity}. In our application we  may for convenience choose
$c_1=\tfrac{\delta_1}{2m}$ and  consider only,  say  $\epsilon\leq
\tfrac{\delta_1}{8m\Delta_1}$. This allows us  to   consider $\tilde C_0 $ as
being independent of  the parameters
   $R^{-1},\kappa,\epsilon >0$ provided they are small.  Fix such a  $\tilde C_0 $.

 {\bf Step III} (Bounding on the complement of $ C_0$.)
 
Note that $\mu(\hat x):= \sigma_s(\hat x) - \eta_0 \cdot
\hat x$ is lower semi-continuous on $S^{d-1}$.   Whence on  any closed   cone $C\subset \mathbb{R}^d
\setminus  C_0$  (for example $C=\mathbb{R}^d
\setminus \tilde C_{0
}$) there exists
\begin{align*}
  \mu_C:=\min \{\mu( x/|x|)|\,0\neq x\in C\}>0,
\end{align*} and for $3 \epsilon<\mu_C$   
another compactness argument shows that $\e^F\phi\in
L^2(C)$. We put this result in a more convenient form: For any smooth function $\chi_C$ on $\R^d$ taken homogeneous of degree zero for
$|x|\geq 1$,   $\chi_C(x) = 0$ in a neighbourhood
of $C_{0}$,  $\chi_C(x) = 0$ for $|x|\leq 1/2$, and with $\chi_C(x) =1$ for
$|x|\geq 1$ and $x$ outside another such neighbourhood, 
\begin{align}\label{eq:elB}
  \sup_n\|\chi_C\phi_n\| <\infty\text{ for small }\epsilon. 
\end{align} 

{\bf Step IV} (Implementation of a scheme from \cite{HS}.)  Consider
the symbol $b_n = r \{X,Y\}$ (the Poisson bracket)
  where $X = \text{Re} (Q(\xi + \i \nabla f_n(x)) - \lambda)$ and $ Y
  =\text{Im} Q(\xi + \i \nabla f_n(x))$. This is given by  \eqref{eq:pois]}.  We
  will freely use other notation from \cite{HS}, in
    particular the  localization symbols  $ \chi_\mp$
      and their   quantizations $\tilde \chi_\mp$  also used
    in   the proof of
  Proposition \ref{just energy} (now with a different $f_n$
 but again in terms of a small parameter $\kappa>0$). Pick any smooth function $\chi_C$ as in Step III with the  property
   that if $\chi_{C_0}^2(x):= 1-\chi_C^2(x)\neq 0$ then either  $|x|< 1$
   or $x\in\tilde C_0$.  Of course we are going to use  \eqref{eq:elB} as well as the lower
   bound of Step II. At
   this point we  can consider the parameters
   $R^{-1},\kappa,\epsilon >0$ of   Steps II and III as fixed (small),
   and with $c:=c_1/3$ we conclude that 
 $(b_n - 3c)\chi_{-}^2
  \chi_{C_{0}}^2 \ge 0$ for $|x|\geq R$.

Noting also  the uniform estimate $(b_n - 3c)\chi_{-}^2
  \chi_{C}^2 \geq -K_1\chi_{C}^2$ we can  then  mimic  \cite[Sections
  6 and 7]{HS}
  and obtain with $A_c:=\Opw (rY)$ and $\tilde{X} +
    \i \tilde{Y}:=Q(p+\i \nabla f_n)-\lambda $: 
\begin{align*}
& 2\text{Im}(A_c(\tilde{X} + \i \tilde{Y}) )\ge \\ 
& 2c +\tilde{Y}r\tilde{Y}  - K_1\chi^2_C- K_2 \tilde{\chi}_+ \langle p\rangle^{2q}\tilde{\chi}_ + +(c - K_3r^{-1/2}\langle p\rangle^{2q}r^{-1/2}) \\
& \ge 2c +\tilde{Y}r\tilde{Y} - K_1\chi^2_C - K_4 \tilde{\chi}_+\langle p\rangle^{2q}\tilde{\chi}_+ - K_5 \chi(r\le N)\langle p\rangle^{2q}\chi(r\le N). 
\end{align*}
In the first  step we used the bound $\chi\langle
  p\rangle^{2q}\chi\leq K r^{-1/2}\langle p\rangle^{2q}r^{-1/2}$ where
  $\chi=\chi\parb{r\leq \sqrt{1+R^2}}$, 
  and in the second  step we used a slightly modified version of \cite[Lemma
4.4]{HS} (applied  with $s=-1/2$, $t=0$ and $\delta=c/K_3$). Taking the expectation in the state $\phi_n$
and  using a slightly modified version of \cite[Lemma 4.3]{HS} we get 
\begin{align}\label{eq:commu}
  \begin{split}
   2c\|\phi_n\|^2  \le & -2\Im \inp{\phi_n,A_cV\phi_n}  - \|r^{1/2}\tilde{Y}\phi_n\|^2 + K_1'\|V\phi_n\|^2  + \\
&  K_1'\|r^{-1/2}\phi_n\|^2 + K_2'\|\langle p\rangle^q\phi\|^2 + K,
 \end{split}
\end{align}
where $K= \sup_n K_1\|\chi_C\phi_n\|^2$.

Taking into account (\ref{assump1}) and (\ref{assump2})  and  the fact
$A_cr^{-1}A_c\leq \tilde{Y}r\tilde{Y}+  K'r^{-1/2}\langle p\rangle^{2q}r^{-1/2}$
(and  by invoking
again \cite[Lemmas 4.3 and 4.4]{HS}) we estimate 
\begin{align*}
&-\i\inp{\phi_n, [V_1, A_c]\phi_n}\\
 &\le \delta\|\langle p\rangle^q \phi_n\|^2  + K_1\|\langle p\rangle^q \chi(r\le N)\phi_n\|^2 + K_2\|\langle p\rangle^q r^{-1/2}\phi_n\|^2 \\ 
&\le \delta' \|\phi_n\|^2 + K_3 \|\langle p\rangle^{q}\phi\|^2,
\end{align*}
 and  
\begin{align*}
&-2\Im \inp{A_c\phi_n,V_2\phi_n} \\
 &\le \|r^{1/2}V_2\phi_n\|^2 + \|r^{-1/2}A_c\phi_n\|^2 \\
&\le \delta\|\phi_n\|^2 + \|r^{-1/2}A_c\phi_n\|^2 + K_4\|\phi\|^2\\
&\le \delta'\|\phi_n\|^2 +\|r^{1/2}\tilde{Y}\phi_n\|^2+ {\color{red}K_5\|\langle p\rangle^{q}\phi\|^2.}
\end{align*}

 We insert these estimates
with $\delta'$ chosen smaller than $c/2$   
into \eqref{eq:commu}  and obtain finally the uniform bound
\begin{equation*}
c\|\phi_n\|^2 \leq \text{ constant},
\end{equation*} accomplishing the goal of Step I.
\end{proof}

\begin{proof} [Proof of Theorem \ref{rademacher1}] There exists a
  sequence $\eta^m=\sigma^m_c(\omega^m)\omega^m$ of regular points in
  $\partial \mathcal{E}^m$ with $\eta^m\to \eta_0=\sigma_0\omega_0$
  for $m\to \infty$, cf. the discussion before Proposition
  \ref{mainthm}. For all elements of this sequence this proposition
  applies and the equations \eqref{maineq:1} and
  \eqref{maineq:2} are satisfied.  Using the ellipticity of $Q$ and by
  going to a subsequence if necessary we can assume $(\xi^m, \theta^m,
  \beta^m, \eta^m) \rightarrow (\xi, \theta,\beta, \eta_0)$ which by
  the continuity of $Q$ and $\nabla Q$ provides a solution to the
  equations \eqref{maineq1} and \eqref{maineq2}. Since
  $\sigma^m_s(\theta^m) = \eta^m \cdot \theta^m$ we can by taking the
  limit show that $\sigma_s(\theta) = \eta_0 \cdot \theta$: For given
  $\epsilon, R>0$ we have for all large $m$ and all $\eta\in\vE^m$
  with $|\eta|\leq R$
  \begin{align*}
    \eta_0\cdot \theta+\epsilon\geq \eta^m\cdot \theta^m\geq
    \eta\cdot \theta^m\geq   \eta\cdot \theta-R \epsilon.
  \end{align*} Taking  $\epsilon\to 0$ yields 
\begin{align*}
    \eta_0\cdot \theta \geq   \eta\cdot \theta \text{ for all
    }\eta\in\vE^m\subset\vE\text{ with }|\eta|\leq R.
  \end{align*}  Then taking $m,R\to \infty$ using \eqref{eq:apprM} and the (related) fact that
  $\sigma^m_c(\omega)=\infty$ if $\sigma_c(\omega)=\infty$ we obtain that $\sigma_s(\theta)
  \leq \eta_0 \cdot \theta$. Obviously $
  \eta_0 \cdot \theta\leq \sigma_s(\theta)$, so $\sigma_s(\theta) =
  \eta_0 \cdot \theta$ is proven.

The second  result follows from the continuity of $1/\sigma_c(\omega)$,
see Theorem \ref{3sigmas}, and the fact that $S^{d-1}$ is connected.
\end{proof}

\section{An example,  $\sigma_{loc} \ne \sigma_s$}\label{sec:example}

In this section we consider for $\epsilon \in (0,1/2)$ the polynomial
\begin{equation*}
Q(\xi) = |\xi|^4+ 2\epsilon \xi_d + \epsilon^2 \xi_d^2
\end{equation*} in dimension $d\geq2$.
A crude estimate gives $Q(\xi) \ge -2\epsilon^{4/3}$.  We take
$\lambda = -1$ so that $\lambda < \inf \sigma(Q(p))$,  and we note 
\begin{align}\label{eq:fac}
  Q(\xi) - \lambda =( |\xi|^2 + \i( 1 +\epsilon \xi_d))( |\xi|^2 - \i( 1
  +\epsilon \xi_d)).
\end{align}

We first solve the system
\begin{align*}
&Q(\xi + \i \sigma \omega) = \lambda, \\
&\nabla Q(\xi + \i \sigma \omega) = \beta \theta
\end{align*}
for $\sigma > 0$ given $\omega \in S^{d-1}$.

The result is that for $\omega_d \neq 0$
\begin{equation} \label{possibilities}
\sigma  = \pm \epsilon \omega_d/2 + \sqrt{\lambda_0^2 - \epsilon^2(1-\omega_d^2)/4},
\end{equation}
where  $2\lambda_0^2 = (\epsilon/2)^2 + \sqrt{ 1 + (\epsilon/2)^4}$.
  
\begin{subequations}
 To see this let  $z=\xi + \i\sigma \omega$.  The system above becomes 
\begin{align*}
z^2 &=\pm \i (1+\epsilon z_d),\\
\beta \theta &= 2(1+\epsilon z_d)(\pm 2\i z + \epsilon e_d).
\end{align*}
A short computation shows that if $\omega_d \ne 0$ then $1+\epsilon z_d \ne 0$ so that redefining $\beta$ we have
\begin{align}
z^2 &=\pm \i(1+\epsilon z_d),  \label{eq1} \\
\beta \theta&= \pm 2\i z + \epsilon e_d. \label{eq2} 
\end{align} 
It is not hard to see that (\ref{eq2}) implies $\mp 2\sigma \omega +
\epsilon e_d = \gamma \xi$ for some $\gamma \in \mathbb{R}$.  Taking
the dot product with $\xi$ and using (\ref{eq1}) gives $\gamma |\xi|^2
= -1$ from which it follows that $\mp 2\sigma \omega + \epsilon e_d =
-\xi/|\xi|^2$.  Combining this equation with the real part of
(\ref{eq1}) gives $(4\sigma \mu + \epsilon^2) \sigma \mu = 1$ where
$\mu = \sigma \mp \epsilon \omega_d$.  A bit more computation yields
\eqref{possibilities}. 
  \end{subequations} 

In addition there is another set of solutions which are only valid for
$\omega_d = 0$. Namely for $d=2$, $\sigma = 1/\epsilon$ and for $d \ge
3$, any $\sigma \ge 1/\epsilon$ independent of $\omega \in S^{d-1}$
such that $\omega_d = 0$. 

Now suppose $(Q(p)+V+1)\phi=0$ for some $V\in C_c^\infty$ and for a
nonzero $\phi\in L^2$. Combining the computation
\eqref{possibilities} with our general results we then conclude that
$0 < \sigma_g < \infty$ and that $\sigma_c(\omega) < \infty$ for all
$\omega\in S^{d-1}$ (note that
near $\omega_d = 0$ the choices (\ref{possibilities}) stay well below
$1/\epsilon$ so  the choice $\sigma \ge 1/\epsilon$ is not relevant).
Thus $\sigma$ must be given by one of (\ref{possibilities}).

Thus there are the following possibilities for \emph{continuous} $\sigma(\omega)$:
\begin{enumerate}[1)]
\item\label{item:a} $\sigma(\omega) = \epsilon \omega_d/2 +
  \sqrt{\lambda_0^2 - \epsilon^2(1-\omega_d^2)/4}$
\item\label{item:b} $\sigma(\omega) = -\epsilon \omega_d/2 +
  \sqrt{\lambda_0^2 - \epsilon^2(1-\omega_d^2)/4}$
\item\label{item:c} $\sigma(\omega) = -\epsilon |\omega_d/2| +
  \sqrt{\lambda_0^2 - \epsilon^2(1-\omega_d^2)/4}$
\item\label{item:d} $\sigma(\omega) = \epsilon |\omega_d/2| +
  \sqrt{\lambda_0^2 - \epsilon^2(1-\omega_d^2)/4}$
\end{enumerate}

The case  \ref{item:d} cannot actually  be $\sigma_c(\omega)$ for the
eigenfunction $\phi$ because it does not describe the boundary of a
convex set. For \ref{item:a} and \ref{item:b} the set $\partial
\mathcal{E}$ is $C^1$ and Corollary \ref{C1} applies. For \ref{item:c}
we cannot apply Corollary \ref{C1} due to the wedge  at $\omega_d=0$
while indeed Theorem \ref{regularpoints} applies near  $\omega_d=\pm
1$ for example.  The sets $\partial\mathcal{E}$ are depicted for the
cases \ref{item:a} and {\color{blue}\ref{item:b}} for $d=2$ by polar plots (this is for
$2\lambda_0/\epsilon=6$ and in terms of the unit $\epsilon/2$):
\vspace{0.5cm}
\begin{center}
\includegraphics[width=9cm,totalheight=7cm]{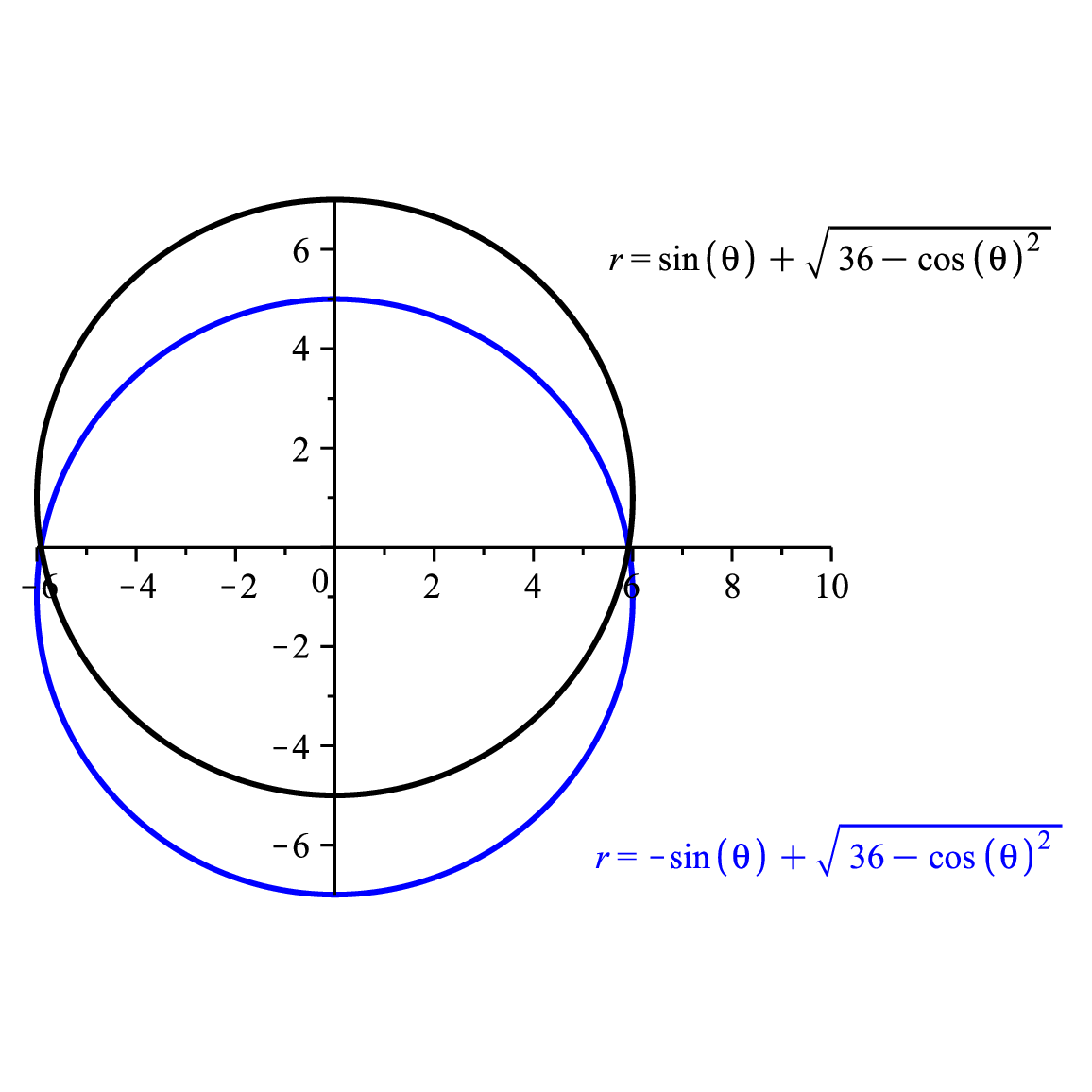}
\end{center}
Note that in this picture $\partial\mathcal{E}$  for the case
\ref{item:c} is the union of the (closed) upper blue arch and the
(closed) lower black
arch, whence there is a wedge at $\omega_d=0$ (in fact defined by 
$\sin \psi=\tfrac\epsilon{2\lambda_0}=\tfrac16$ where  $\pi-2\psi$ is the
apex angle).
In general   a computation for case \ref{item:c} shows that Theorem
\ref{regularpoints} applies  if and only if $ |\omega_d| >
\tfrac\epsilon{2\lambda_0}$ and in this case
\begin{align*}
  \sigma_{loc}(\omega)=\sigma_s(\omega) =\lambda_{0}-\tfrac\epsilon{2}|\omega_d|.
\end{align*}
If on the other hand $ |\omega_d| \leq 
\tfrac\epsilon{2\lambda_0}$ we compute for case \ref{item:c}
\begin{align}\label{eq:sigSy}
  \sigma_s(\omega) =(1 - \omega_d^2)^{1/2} \sqrt{\lambda_0^2 -
    \epsilon^2/ 4}.
\end{align} We present below an example of case \ref{item:c} where
$\sigma_{loc}(\omega)> \sigma_s(\omega)$ for $ |\omega_d| <
\tfrac\epsilon{2\lambda_0}$. 

Let us first note that there are examples of \ref{item:a} and
\ref{item:b}: Indeed (motivated by \eqref{eq:fac}) we  take
\begin{align*}
  \phi_{\pm}=\chi+(1-\chi)( p^2 \pm \i( 1 +\epsilon p_d))^{-1}\delta,
\end{align*} where $\chi\in C_c^\infty$, $0\leq \chi\leq1$, $\chi$ is $1$ in a small neighbourhood
of $0$ and 
has small support,  and $\delta$ is the delta function at $0$. If
$g_{\pm} (x)$ denotes the Green's function $(p^2\pm
\i+\epsilon^2/4)^{-1}(x,0)$ then 
\begin{align}
  \label{eq:ei}
  \phi_{\pm}(x)=\chi(x)+\parb{1-\chi(x)}\e^{\pm x_d\epsilon/2}g_{\pm} (x).
\end{align} Using properties of $g_{\pm} (x)$ (see the discussion in
\cite[Subsection 1.2]{HS} and note that $\sqrt{\mp
  \i-\epsilon^2/4}=\i\lambda_0\mp (2\lambda_0)^{-1}$) we deduce  that  each choice $\phi= \phi_\pm$ fulfills
$(Q(p)+V+1)\phi=0$ for some $V\in C_c^\infty$. The choice $\phi_-$ is
an example of the case \ref{item:a} while the  choice $\phi_+$ is
an example of the case \ref{item:b}. In general for these cases we
have  that for all $\omega\in S^{d-1}$
\begin{align*}
  \sigma_{loc}(\omega)=\sigma_s(\omega) =\lambda_{0}\mp\tfrac\epsilon{2}\omega_d,
\end{align*} respectively.

Now for an example of case \ref{item:c}, we consider
\begin{equation*}
g(x) = (2\pi)^{-d}\int \e^{\i x\cdot
  \xi}(1+\epsilon \xi_d)(Q(\xi) +1)^{-1}  d\xi;\, x\neq 0.
\end{equation*} It is well-defined and smooth, and  introducing as above
\begin{align}
  \label{eq:ei0}
  \phi(x)=\chi(x)+\parb{1-\chi(x)}g (x),
\end{align}  for this $\phi$ we have 
$\sigma_g\in(0,\infty)$, cf. Paley-Wiener theory.
We claim that $(Q(p)+V+1)\phi=0$ for some $V\in
C_c^\infty$ and that  this is an example of case \ref{item:c}. Since
$g(x)=\overline{g(-x)} $ we have $\sigma_c(\omega) =
\sigma_c(-\omega)$ is  valid for all $\omega$. This excludes   the cases
\ref{item:a} and \ref{item:b} and we are left with case
\ref{item:c}. Whence it remains to construct $V$. We use the function $g_{-} (x)=(p^2-
z)^{-1}(x,0)$, $z=\i-\epsilon^2/4$, from above and represent
\begin{align*}
  2\i g(x)=\e^{- x_d\epsilon/2}g_{-} (x)-\e^{ x_d\epsilon/2}\overline {g_{-}
  (x)}\mand \Re g(x)=\cosh (x_d\epsilon/2)\Im g_{-} (x).
\end{align*} Next we use that $\Im \parb{(p^2-
z)^{-1}(x,0)}>0$ for all $|x|>0$ small enough (this is valid for any
$d\geq 2$
and for any
$z\in\C$ with $\Im z>0$).  For example in dimension  $d=3$  explicitly for
$z=\i-\epsilon^2/4$ this property holds for 
$|x|< \pi(2\lambda_0)^{-1}$. Whence by possibly
adjusting the support of $\chi$ we can safely define
\begin{align}\label{eq:pot}
  V=-\{(Q(p) +1)\phi\}/\phi\in C_c^\infty.
\end{align}

Finally from the asymptotics of $g_{-}$ we obtain
$\sigma_{loc}(\omega)=\lambda_{0}-\tfrac\epsilon{2}|\omega_d|$. 
Comparing with \eqref{eq:sigSy}  we see that for the
eigenfunction \eqref{eq:ei0}  indeed 
$\sigma_{loc}(\omega)> \sigma_s(\omega)$ for $ |\omega_d| <
\tfrac\epsilon{2\lambda_0}$.

\begin{remarks*} It is easy to check that the potential $V$ of
  \eqref{eq:pot} satisfies $\overline{RV}=V$ where $Rf(x_{\perp},
  x_d)=f(x_{\perp}, -x_d)$ using the fact that also $Q(p)$ has conjugate
  reflection symmetry. However there is no reason to believe that $V$
  is real-valued. If  on the other hand we pick an arbitrary real
  nonzero $V\in C_c^\infty$, $V\geq 0$, the variational principle
  shows that for some $\kappa<0$ the energy $\lambda=-1$ is an
  eigenvalue of $H=Q(p)+\kappa V$. If furthermore $RV=V$ then we can pick a corresponding
  eigenfunction $\phi$ obeying $\overline{R\phi}=\phi$. This $\phi$ is
  an example of case \ref{item:c} with  a real potential in $
  C_c^\infty$. However it appears  difficult to compute asymptotics
  for $ |\omega_d| \leq 
\tfrac\epsilon{2\lambda_0}$. We
  claim that for $d=3$ at least  $\sigma_{loc}(\omega)>
  \sigma_s(\omega)$ when $\omega_d=0$. This can be done by first
  representing the Green's function without potential as
  \begin{align*}
    \parb{Q(p)+1}^{-1}(x,0)=e^{- x_d\epsilon/2}\int g_{-} (x-y)
    e^{y_d\epsilon}g_{+} (y) dy,
  \end{align*} where $g_{\pm}$ are given as above. For $d=3$ we may
  use the familiar expression $(p^2- z)^{-1}(x,0)=(4\pi)^{-1}\e^{\i\sqrt
    z|x|}/|x|$, $\Im \sqrt z>0$, and estimate this integral explicitly
  (after a suitable deformation of contour) and  show that indeed $\sigma_{loc}(\omega)> \sqrt{\lambda_0^2 -
    \epsilon^2/ 4}$ when
  $\omega_d=0$. We skip the details.
\end{remarks*}

\section{The Agmon metric and a variational principle}\label{sec:Agmon}

Here we discuss some connection to previous works \cite{Ag1, Ag2}
which applies for example to the case \ref{item:c}  of Section
\ref{sec:example}.  As we will see we are not going to derive better
bounds than we already have.    Our analysis applies to an eigenvalue $\lambda$  not in
$\text{Ran}Q$ and results in a set $\mathcal{E}_A$ whose
boundary is described by the same equations as the boundary of the
sets $\mathcal{E}$ which we have seen above.  In other words the set
$\partial\mathcal{E}_A$ is just a subset of the solutions of the
equations (\ref{maineq1}) and (\ref{maineq2}).  For the  case
\ref{item:c} of Section \ref{sec:example} the boundary 
$\partial\mathcal{E}=\partial\mathcal{E}_A$,  however this is 
not valid for the cases \ref{item:a} and \ref{item:b}.  As an
additional bonus we will see that quite generally all the solutions to (\ref{maineq1}) and (\ref{maineq2}) can
 be obtained from the same  variational principle which we use to
 derive the equations satisfied by the points of
 $\partial\mathcal{E}_A$.  In the following we assume $\phi$ is an
 eigenfunction of $H$,  $(H-\lambda)\phi =0$.  We assume $V =o(1)$ at infinity and $\lambda \notin \text{Ran}Q$.

In analogy with what Agmon does \cite{Ag2} for the Laplacian, we consider the set of all real-valued $f \in C^1(\mathbb{R}^d)$ such that 
\begin{equation}\label{Agmonest}
\|(Q(p + \i \nabla f(x)) + V(x) - \lambda)\psi\| \ge \delta \|\psi\|
\end{equation}
with $\psi \in C_c^{\infty}(\mathbb{R}^d\setminus{\bar B_R}), B_R = B_R(0)$ for some $R$ and positive $\delta$.  Since our $V$ is $o(1)$ at infinity, it can be omitted from (\ref{Agmonest}) and we get an equivalent estimate.  We mention that the quadratic form estimate of Agmon in the case of a second order operator implies (\ref{Agmonest}).

Let   

\begin{equation}\label{Agmonest2}
\delta(f) = \lim_{R \to \infty} \inf \{\|(Q(p + \i \nabla f(x)) - \lambda)\psi\| \ | \  \psi \in C_c^{\infty}(\mathbb{R}^d\setminus{\bar B_R)}, \|\psi\| = 1\}.
\end{equation}
Note that $\delta(f)$ is invariant under translations:  $\delta(f) = \delta(f_a)$, $f_a(x) = f(x-a)$.  Thus $\delta(f)$ depends on the values of $\nabla f(x)$ for large $x$ rather than for what $x$ they are taken on.  From the viewpoint of using psdo's to get an estimate such as (\ref{Agmonest2}) with positive $\delta(f)$ it is natural to look at $f$'s which are symbols of order $1$ for which $ \nabla f(x)$   is in the set
\begin{align*}
  \mathcal{E}_A : = \{\eta \in \mathbb{R}^d | Q(\xi + \i t\eta) - \lambda \ne 0 \ \forall \xi \in \mathbb{R}^d,\ \forall t \in [0,1]\}
\end{align*}
   for all large $x$.  We do not show that this set of $f$ satisfy an estimate such as (\ref{Agmonest2}) (although this can be done) but rather come at the question from a different point of view.  We do mention an important reason for assuming that all smaller values of $\nabla f(x)$ in the same direction be in the set (the reason for $t$ in the definition of $\mathcal{E}_A$).  This is automatic with Agmon's quadratic form estimate but more importantly when trying to prove an estimate such as $\e^f\phi \in L^2$ one needs to approximate $f$ with smaller functions $f_{\epsilon}$ for which one knows apriori that $\e^{f_{\epsilon}}\phi \in L^2$.

Let $k$ be the Minkowski functional, $k(\eta) = \inf\{t>0| \eta/t \in
\mathcal{E}_A\}$, of the bounded convex 
open set $\mathcal{E}_A$.  (The convexity follows from \cite{Ho}.)
It follows that $\mathcal{E}_A = \{\eta | k(\eta) <
1\}$. Following Agmon \cite{Ag1} we introduce the polar $k_*(x) =
\sup\{x\cdot \eta/k(\eta)| \eta \ne 0\} = \sup\{x\cdot \eta| \eta \in
\mathcal{E}_A\}$.  $k_*$ is just the support  function of the
bounded convex set ${\bar{\mathcal{E}}_A}$.  Finally the Agmon metric based on $\mathcal{E}_A$ is 
\begin{align*}
  \rho_A(x,y) = \inf \{\int_0^1 k_*(\dot{\gamma}(t)) dt \ | \gamma(0) = y, \gamma(1) = x, \gamma(\cdot) \ \text{absolutely continuous}\}.
\end{align*}
Note that from Theorem \ref{thm:start} and Proposition \ref{just
  energy} it follows that each $\eta \in \mathcal{E}_A$ satisfies $
\e^{\eta\cdot x} \phi \in L^2$ and thus $\e^{tk_*(x)} \phi \in L^2$  for all $t \in [0,1)$.  (This can also be shown using the Combes-Thomas method \cite{CT}.)  We claim that
actually $\rho_A (x,0) = k_*(x)$.  First note that $\rho_A(x,0) \le
\int_0^1k_*(\dot{\gamma}(t))dt$ where $\gamma(t) = tx$.  This gives
$\rho_A(x,0) \le k_*(x)$. The opposite estimate, $\rho_A(x,0)
  \geq k_*(x)$, follows readily from the fact that $k_*$ is  a norm. We
give in the following a more informative proof, although it is  more
complicated. Let  $x$ be  a point of differentiability of
$k_*(x)$ (since $k_*$ is convex it is differentiable a.e.).  Pick  a point $\eta
\in \partial\mathcal{E}_A$ with $k_*(x) = \eta \cdot x$.  Then
$k_*(x+s\omega) \geq \eta \cdot (x + s \omega) = k_*(x) + s \eta \cdot
\omega$ and thus taking $s > 0 $ and then $s \to 0$ we obtain $\nabla
k_*(x) \cdot \omega  \geq \eta\cdot \omega$.  Since $\omega \in
\mathbb{R}^d$ is arbitrary we obtain  $\nabla k_*(x) =\eta$.  By
definition of $k$ this implies $k(\nabla k_*(x)) \le 1$ which 
by   \cite[Lemma 1.3]{Ag2}  implies  
 $k_*(x) \le \rho_A(x,0)$.  Thus $k_*(x) = \rho_A(x,0)$. Since we already knew that  $\e^{tk_*(x)} \phi \in L^2$ for all $t \in [0,1)$ 
the Agmon bound, $\e^{t\rho_A(x,0)}\phi \in L^2$ for $t \in [0,1)$, gives no new information.

\textbf{The variational principle:
}
We now turn to finding equations describing the set $\partial
\mathcal{E}_A$.  We fix $\omega_0 \in S^{d-1}$ and attempt to find the
minimum value, say $\sigma_0$, of $\sigma>0$ such that $Q(\xi + \i \sigma \omega_0)
=\lambda$ for some $\xi$.  We are of course still in the situation where $\lambda \not\in
\text{Ran}Q$. The point $\eta_0:=\sigma_0 \omega_0$ will then be in $\partial
\mathcal{E}_A$.  We want to use Lagrange multipliers.  For this
purpose define the two functions $f_1(\xi, t) = \text{Re}Q(\xi +
\i t\omega_0)$ and $f_2(\xi, t) = \text{Im}Q(\xi + \i t\omega_0)$.  If these functions are independent at a minimum point $(\xi_0 ,\sigma_0)$ in the sense that the two gradients $\nabla_{(\xi, t)}f_j(\xi_0, \sigma_0)$ are linearly independent then defining the function $F(\xi, t) = t + \alpha_1 f_1(\xi, t) + \alpha_2 f_2(\xi, t)$ and setting the derivatives equal to zero gives
\begin{subequations}
\begin{align}\label{Lagrange1}
1 - \alpha_1 \omega_0 \cdot \text{Im}\nabla Q(\xi_0 +\i\eta_0) + \alpha_2 \omega_0 \cdot \text{Re}\nabla Q(\xi_0 +\i \eta_0) = 0, &\\
\alpha_1 \text{Re}\nabla Q(\xi_0 +\i \eta_0) + \alpha_2 \text{Im}\nabla Q(\xi_0 +\i\eta_0) =0. \label{Lagrange2}
\end{align}
  \end{subequations}
Evidently $\alpha_1^2 + \alpha_2^2 > 0$ so that the real and imaginary parts of $\nabla Q(\xi_0 +\i \eta_0)$ are linearly dependent which means that for some $\beta \in \mathbb{C}$ and $\theta \in S^{d-1}$ we have 
\begin{subequations}
\begin{align} \label{usualeqs1}
Q(\xi_0 + \i \eta_0) &= \lambda, \\
\nabla Q(\xi_0 +\i \eta_0)& = \beta \theta. \label{usualeqs2}
\end{align}
\end{subequations}

On the other hand if the gradients of $f_1$ and $f_2$ are linearly dependent  at the point $(\xi_0, \sigma_0)$ then again the real and imaginary parts of $\nabla Q(\xi_0 +\i \eta_0)$ are linearly dependent at this point and the equations (\ref{usualeqs1}) and (\ref{usualeqs2}) hold. 

Conversely consider a point $\eta_0 = \sigma_c(\omega_0) \omega_0 \in \partial\mathcal{E}$ where we have in mind an eigenfunction $\phi$ of $H = Q(p) + V(x)$  with $\mathcal{E} = \{\eta \in \mathbb{R}^d | \e^{\eta \cdot x}\phi \in L^2\}$.  Suppose the eigenvalue $\lambda$ is not a critical value of $Q(z)$ but we no longer assume that  $\lambda \notin \text{Ran}Q$.  
There are $\xi_0 \in \mathbb{R}^d, \beta \in \mathbb{C}\setminus
\{0\}$ and $\theta \in S^{d-1}$ such that $\sigma_s(\theta) =
\eta_0\cdot \theta$ and such that these quantities along with
$\sigma_0 = \sigma_c(\omega_0)$ satisfy equations (\ref{usualeqs1})
and (\ref{usualeqs2}).  We claim the the gradients of the functions
$f_j$ are linearly independent at this point.   If for real $\alpha_1$
and $\alpha_2$ not both zero we have $\alpha_1 \nabla_{(\xi,t)}
f_1(\xi_0, \sigma_0) + \alpha_2 \nabla_{(\xi,t)} f_2(\xi_0, \sigma_0)
= 0$, we calculate $\alpha_1 \nabla_{\xi} f_1 + \alpha_2 \nabla_{\xi}
f_2 = 0$  and $\omega_0 \cdot (-\alpha_1 \nabla_{\xi} f_2 + \alpha_2
\nabla_{\xi} f_1) =0$.  These two equations imply $\nabla_{\xi}
f_1\cdot \omega_0 = \nabla_{\xi} f_2\cdot \omega_0 =0$ or $\nabla
Q(\xi_0 + \i \eta_0) \cdot \omega_0 = 0$.  Thus since
$\lambda$ is not a critical value of $Q(z)$, $\omega_0 \cdot \theta
=0$.  This contradicts the geometry of $\partial \mathcal{E}$:  Since
$\bar {\mathcal{E}}\subset \{\eta | (\eta - \eta_0)\cdot \theta \le
0\}$ and since for example $\eta_0/2$ is an interior
point of $\bar {\mathcal{E}}$  we can take $\eta =
\eta_0/2 + u$ above where $u$ is small and learn that
$u\cdot \theta \le 0$ for all small $u$, a contradiction.  Thus for
some small $\epsilon >0$, $\{(\xi,t) | Q(\xi + \i t \omega_0) = \lambda, |\xi - \xi_0|+ |t-\sigma_0| < \epsilon\}$ is a co-dimension two smooth submanifold of $\mathbb{R}^{d+1}$ and $(\xi_0, \sigma_0)$ is a critical point of the function $F(\xi, t) = t$ restricted to this submanifold.  This is because given the equations (\ref{usualeqs1}) and (\ref{usualeqs2}) we can find $\alpha_1$ and $\alpha_2$ solving the equations (\ref{Lagrange1}) and (\ref{Lagrange2}).  Thus the equations of Theorem \ref{rademacher1} coming from an eigenfunction of $H$ can be derived from this variational principle.

\subsection{The set $\bar {\mathcal{E}}$ for the Green's function}\label{Green's function}

Suppose $\lambda \in \mathbb{R} \setminus \text{Ran} Q$. Then the exponential decay of the Green's function
$$G(x-y) = (Q(p) - \lambda)^{-1}(x,y)$$
is naturally associated with the set 
$$\mathcal{E}_G =\{\eta \in \mathbb{R}^d | \e^{\eta\cdot x}G \in L^2(\mathbb{R}^d \setminus  B_1(0) )\}$$
whose boundary points, $\eta_0$, are associated  to
  solutions of the equations (\ref{usualeqs1}) and (\ref{usualeqs2}).   Clearly $\bar{\mathcal{E}}_A \subset \bar{\mathcal{E}}_G$.  In fact there is equality as is easy to show:  Suppose $\eta_0 = t_0 \omega_0 \in \bar{\mathcal{E}}_G \setminus\bar{\mathcal{E}}_A$, where $0 < t_0$ and  $\omega_0 \in S^{d-1}$.  We write
 \begin{equation} \label{eqforres}
 (Q(\xi) - \lambda)^{-1} = (2\pi)^{-d}\int _{\mathbb{R}^d \setminus B_1(0)} e^{\i x \cdot \xi} G(x) dx + F(\xi)
\end{equation}
where $F$ is an entire function.  Since $\sigma_g > 0$ we can use the convexity of  $\bar{\mathcal{E}}_G$ and the continuity of $1/\sigma_c(\omega)$ to show that given $\epsilon >0$ there is a $\delta>0$ so that for all $\eta$ in a neighbourhood of $\{t\omega_0 | 0 < t < t_0 - \epsilon\} $
$$e^{\eta \cdot x + \delta |x|} G \in L^2(\mathbb{R}^d \setminus  B_1(0) ).$$
Using the Cauchy-Schwarz inequality in (\ref{eqforres}) this implies that $(Q(\xi) - \lambda)^{-1} $ has an analytic continuation to a neighbourhood of 
$$\{z = \xi -\i t \omega_0 | \xi \in \mathbb{R}^d, 0 < t < t_0 - \epsilon\}.$$
But for small $\epsilon$ the zero set of $ Q(z) - \lambda$ will intersect this neigh{bour}hood, a contradiction.  

To end this subsection perhaps it is worth emphasizing a fact that we used above:  The set $\bar{\mathcal{E}}$ for any eigenfunction contains $\bar{\mathcal{E}}_A = \bar{\mathcal{E}}_G$ at the same spectral parameter.

\section{The set $\bar{\mathcal{E}}$ - smoothness}\label{sec:set-barmathcale}

The set of $\eta$ satisfying
\begin{subequations}
\begin{align}
    Q(\xi+\i \eta)&=\lambda, \label{main1}\\ 
    \nabla Q(\xi +\i \eta)&= \beta \theta. \label{main2}
  \end{align}
\end{subequations}
for some $\xi, \beta, \theta$ is a semi-algebraic set (see
\cite{BCR}).  By definition this means that it is a finite union of
sets of the form $S_n = \{\eta \in \mathbb{R}^d | q_j(\eta) = 0,
p_j(\eta) > 0, j = 1, \cdots, n\}$ where 
  the $p_j$ and $q_j$ are real polynomials. This comes from the fundamental result that a projection of a semi-algebraic set is a semi-algebraic set.  It would be interesting to know what restrictions this puts on the set of singular points of the boundary of the set $\mathcal{E}$ defined for an eigenfunction of $H$ with $0 < \sigma_g < \infty$.

We give sufficient conditions for the local smoothness of solutions, $z = \xi + \i\eta = h(\theta)$,  of (\ref{main1}) and  (\ref{main2}).  We do not assume that solutions come from the exponential decay of an eigenfunction of $Q(p) + V(x)$.  Let us assume $\lambda$ is not a critical value of $Q$ so that given a solution $(\xi_0, \eta_0, \beta_0, \theta_0)$, $\beta_0$ must be nonzero.  Let us assume $Q''(z_0)$ is invertible ($z_0 = \xi_0 + \i \eta_0$).  Generically this is true when $Q(z_0) = \lambda$ except on a $d-2$ dimensional  manifold.  Then we can define locally the Legendre transformation $P(w) = z\cdot w - Q(z),  w = \nabla Q(z)$.  $\nabla P$ is the inverse of $\nabla Q$.  We then have $\nabla P(\beta \theta) = z = \xi + \i \eta$ so that $Q(\nabla P(\beta \theta)) = \lambda$.  We can solve for $\beta$ in terms of $\theta$ locally if $\frac{\partial}{\partial \beta} Q(\nabla P(\beta \theta)) \ne 0$ when $\beta = \beta_0, \theta = \theta_0$.  A short calculation gives the requirement $\beta_0 \theta_0\cdot P''(\beta_0\theta_0) \theta_0 \ne 0$ so that we have:
 \begin{proposition}
Suppose $Q(z_0) = \lambda$, $ \nabla Q(z_0) = \beta_0 \theta_0$, $\beta_0\neq 0$,
$Q''(z_0)$ is invertible and 
\begin{equation*}
\nabla Q(z_0) \cdot Q''(z_0)^{-1}\nabla Q(z_0) \ne 0.
\end{equation*} 
Then there exists a neighbourhood of $(\theta_0 , \beta_0 , z_0)$ in
which the set of solutions to the system \eqref{main1} and
\eqref{main2} (with $z=\xi+\i\eta$) is parametrized smoothly by $\theta$.
\end{proposition}

Given the assumptions of the proposition, we have $z = h(\theta)$ in a
neighbourhood of $(z_0, \theta_0)$.  We can calculate the derivative
$h'(\theta)$ by differentiating $ \nabla Q(z) = \beta \theta$ and using the formula for $\beta'(\theta)$ from the above
application of the implicit function theorem.   We obtain as an
identity on the tangent space 
$T_\theta( S^{d-1})=\{x\in\R^{d}|\, x\cdot \theta=0\}$ 
\begin{align*}
  h'(\theta) = \beta Q''(z)^{-1}(I - R(\theta));\,R(\theta)x := \frac{\theta\cdot Q''(z)^{-1}x}{\theta\cdot
  Q''(z)^{-1}\theta}\theta.
\end{align*}

To better understand the meaning of the relationship between $\eta$
and $\theta$ let us take $\theta_1$ near $\theta_0$ and look for a
critical point of the function $\eta \cdot \theta_1$ for $\eta =
\text{Im}h(\theta)  = : g(\theta)$.  Since
  $\theta\cdot h'(\theta)=0$ by the above formula obviously $\theta=\theta_1$ is
a  critical point of the function $g \cdot \theta_1$.  This is consistent with the geometric interpretation of $(\theta_1,g(\theta_1))$ being the parameters of a supporting hyperplane at the boundary point $\eta_1 = g(\theta_1)$ of the convex set $\mathcal{E}$ which comes from an $L^2$-function $\phi$ solving $(H-\lambda)\phi = 0$.  In this case $\eta \cdot \theta_1$ would be maximized with $\eta = \eta_1$.  If $\eta = g(\theta)$ describes the boundary of a convex set $\mathcal{E}$ which comes from an $L^2$-function $\phi$ solving $(H-\lambda)\phi = 0$ then the uniqueness of $\eta$ corresponds to the strict convexity of $ \bar{\mathcal{E}}$. 

The conditions which allow us to conclude that $\eta$ is a smooth
function of its direction $\omega = \eta/|\eta|$ are more
complicated.  If $g(\theta) = \text{Im} h(\theta)$ as above, we want to solve
for ($\sigma,\theta)$ as a function of $\omega$ in the equation
$\sigma \omega - g(\theta) = 0$ near $ (\sigma\omega, \theta) =
(\eta_0, \theta_0)$.  The inverse function theorem gives the result
that $\sigma$ is locally a smooth function of $\omega$ if the only
solution $(x, \mu)$ 
to the real linear equation $g'(\theta_0)x = \mu \eta_0$ is the trivial solution.
 Let us make the assumption that $g(\theta_0)\cdot
 \theta_0\neq 0$ and  the (generic) assumption $\ker g'(\theta_0) = {0}$. Note that if the equation $\eta = g(\theta)$  represents the boundary of
a set $\mathcal{E}$ which comes from a solution $\phi$ to
$(H- \lambda)\phi = 0$ with  $\sigma_g>0$ and  $\sigma_s(\theta) =
g(\theta)\cdot \theta$, cf. Theorem \ref{rademacher1},  then
obviously $g(\theta)\cdot \theta\neq 0$.  Now
differentiating $Q(z) = \lambda$ gives $\theta_0 \cdot h'(\theta_0) =
0$ and therefore also that $0=\theta_0 \cdot g'(\theta_0) x=\mu \theta_0 \cdot \eta_0= \mu g(\theta_0)\cdot
 \theta_0$
 showing that $\mu = 0$ and then in turn  $x=0$. Whence  the only 
 solution to $g'(\theta_0)x = \mu \eta_0$ is the trivial one.
 
 \textbf{Acknowledgement:}  E. S. was supported by Grant 11-106598 FNU.

\end{document}